\documentclass[11pt,letterpaper]{article}

\usepackage{graphicx, titling}

\usepackage{amsmath,amsthm,amsfonts,amssymb,color,hyperref,anysize,enumerate,graphicx,epstopdf,cleveref,multirow}
\usepackage[usenames,dvipsnames]{xcolor}
\usepackage[margin=1.25in]{geometry}

\newtheorem{theorem}{Theorem}
\newtheorem{lemma}{Lemma}

\newtheorem{claim}[lemma]{Claim}
\newtheorem{assume}{Assumption}

\newtheorem{example}{Example}

\newcommand{\ra}{\rightarrow}

\DeclareMathOperator*{\argmax}{arg\,max}
\DeclareMathOperator{\sign}{sign}

\usepackage{tcolorbox}

\newcommand{\rjc}[1]{{\color{blue}{#1}}}

\title{Balancing the Robustness and Convergence of Tatonnement
\thanks{This work was supported in part by NSF grants CCF-1527568 and CCF-1909538.}
}
\author{Richard Cole \\ 		Courant Institute, NYU
\and
Yixin Tao
\\		Courant Institute, NYU
}
\date{}

\begin{document}
\newcommand{\hide}[1]{}
\newenvironment{pfof}[1]{\begin{proof}[\emph{\textbf{Proof of #1: }}]}{\end{proof}}

\newcommand{\largespace}{\hspace*{1.0in}}
\newcommand{\smallspace}{\hspace*{0.5in}}

\newcommand\numberthis{\addtocounter{equation}{1}\tag{\theequation}}

\newcommand{\bba}{\mathbf{a}}
\newcommand{\bbb}{\mathbf{b}}
\newcommand{\bbc}{\mathbf{c}}
\newcommand{\bbd}{\mathbf{d}}
\newcommand{\bbe}{\mathbf{e}}
\newcommand{\bbf}{\mathbf{f}}
\newcommand{\bbg}{\mathbf{g}}
\newcommand{\bbh}{\mathbf{h}}
\newcommand{\bbi}{\mathbf{i}}
\newcommand{\bbj}{\mathbf{j}}
\newcommand{\bbk}{\mathbf{k}}
\newcommand{\bbl}{\mathbf{l}}
\newcommand{\bbm}{\mathbf{m}}
\newcommand{\bbn}{\mathbf{n}}
\newcommand{\bbo}{\mathbf{o}}
\newcommand{\bbp}{\mathbf{p}}
\newcommand{\bbq}{\mathbf{q}}
\newcommand{\bbr}{\mathbf{r}}
\newcommand{\bbs}{\mathbf{s}}
\newcommand{\bbt}{\mathbf{t}}
\newcommand{\bbu}{\mathbf{u}}
\newcommand{\bbv}{\mathbf{v}}
\newcommand{\bbw}{\mathbf{w}}
\newcommand{\bbx}{\mathbf{x}}
\newcommand{\bby}{\mathbf{y}}
\newcommand{\bbz}{\mathbf{z}}

\newcommand{\bbA}{\mathbf{A}}
\newcommand{\bbB}{\mathbf{B}}
\newcommand{\bbC}{\mathbf{C}}
\newcommand{\bbD}{\mathbf{D}}
\newcommand{\bbE}{\mathbf{E}}
\newcommand{\bbF}{\mathbf{F}}
\newcommand{\bbG}{\mathbf{G}}
\newcommand{\bbH}{\mathbf{H}}
\newcommand{\bbI}{\mathbf{I}}
\newcommand{\bbJ}{\mathbf{J}}
\newcommand{\bbK}{\mathbf{K}}
\newcommand{\bbL}{\mathbf{L}}
\newcommand{\bbM}{\mathbf{M}}
\newcommand{\bbN}{\mathbf{N}}
\newcommand{\bbO}{\mathbf{O}}
\newcommand{\bbP}{\mathbf{P}}
\newcommand{\bbQ}{\mathbf{Q}}
\newcommand{\bbR}{\mathbf{R}}
\newcommand{\bbS}{\mathbf{S}}
\newcommand{\bbT}{\mathbf{T}}
\newcommand{\bbU}{\mathbf{U}}
\newcommand{\bbV}{\mathbf{V}}
\newcommand{\bbW}{\mathbf{W}}
\newcommand{\bbX}{\mathbf{X}}
\newcommand{\bbY}{\mathbf{Y}}
\newcommand{\bbZ}{\mathbf{Z}}

\newcommand{\aij}{a_{ij}}
\newcommand{\aik}{a_{ik}}
\newcommand{\bij}{b_{ij}}
\newcommand{\bik}{b_{ik}}
\newcommand{\bijt}{b_{ij}^t}
\newcommand{\bikt}{b_{ik}^t}
\newcommand{\cij}{c_{ij}}
\newcommand{\bps}{\mathbf{p}^*}
\newcommand{\pj}{p_{j}}
\newcommand{\ps}{p^*}
\newcommand{\rhoi}{{\rho_i}}
\newcommand{\xij}{x_{ij}}
\newcommand{\xj}{x_{j}}
\newcommand{\xijt}{x_{ij}^t}
\newcommand{\xikt}{x_{ik}^t}
\newcommand{\xt}{x^{t}}
\newcommand{\xto}{x^{t+1}}
\newcommand{\zj}{z_j}

\newcommand{\hsp}{\hspace*{0.1in}}
\newcommand{\innerprod}[2]{\langle #1,#2\rangle}
\newcommand{\rr}{\mathbb{R}}
\newcommand{\rrplusn}{\mathbb{R}_+^n}

\begin{titlingpage}
        \maketitle
  \begin{abstract}
A major goal in Algorithmic Game Theory is to justify equilibrium concepts from
an algorithmic and complexity perspective.
One appealing approach is to identify robust natural distributed algorithms
that converge quickly to an equilibrium.
This paper addresses a lack of robustness in existing convergence results  
for discrete forms of tatonnement,
including the fact that it need not converge when buyers
have linear utility functions.
This work achieves greater robustness by seeking approximate 
rather than exact convergence
in large market settings.

More specifically, this paper shows that for Fisher markets with buyers having
CES utility functions, including linear utility functions,
tatonnement will converge quickly to an approximate equilibrium
(i.e.\ at a linear rate), modulo a suitable large market assumption.
The quality of the approximation is a function of the
parameters of the large market assumption.
\end{abstract}
\end{titlingpage}

\newpage
\thispagestyle{plain}\setcounter{page}{1}

\section{Introduction}
\label{sec::intro}

To show the plausibility of equilibrium concepts one would like 
simple, robust procedures
that quickly reach or at least approach an equilibrium state.
But, as is well known, it is PPAD-hard to compute equilibria for
general economies~\cite{chen2013complexity, garg2017settling, chen2009settling}.
Consequently (assuming no unexpected complexity results
such as PPAD = FP) there are no polynomial algorithms
to compute an equilibrium in markets in general,
let alone simple, robust, and rapidly convergent procedures.

As a result, considerable attention has been given to the design
of polynomial time algorithms to find equilibria for specific families of
economies~\cite{garg2018strongly, duan2015combinatorial, ghiyasvand2012simple, jain2007polynomial, orlin2010improved, shmyrev2009algorithm, ye2008path}, and also to the analysis of simple dynamic processes,
most notably tatonnement~\cite{cole2008fast, cheung2019tatonnement, cheung2012tatonnement, cheung2018amortized} 
and proportional response~\cite{birnbaum2011distributed, zhang2011proportional, wu2007proportional, cheung2018dynamics}.
One class of economies that has received considerable attention
in the computer science literature
are Fisher markets~\cite{devanur2004spending},\footnote{In the CS literature the term market has been widely used to refer to economies;
we follow this practice.} which generalize the equal income property
of the CEEI setting~\cite{varian74} to arbitrary incomes.

This paper focuses on discrete versions of tatonnement.
Recall that the tatonnement update rule increases
the price of a good when its demand is too high,
and reduces it when the demand is low.
It is well known that tatonnement need not converge
when buyer utilities are linear as shown in the following simple example.

\begin{example}
\label{ex::tat-non-conv}
There are two items, both with unit supply, and one buyer with 2 units of money whose 
utility equals the sum of the amount of the two items it receives.
Suppose we use the update rule $p'_j=p_j\exp(\lambda \min\{x_j-1,1\})$,
where  $x_j$ is the
demand for good $j$, and $\lambda > 0$ is a parameter; 
this is essentially the version of the tatonnement rule
we will consider in this paper, and a version that has been analyzed previously.
Suppose the prices for the two items
are initially $p_1 = e^{\lambda/2}$ and $p_2=e^{-\lambda/2}$, respectively.
Then the demand for good 1 is 0 and the demand for good 2 is
$2e^{\lambda/2}$, so following one round of updates, the prices
become $p_1 = e^{-\lambda/2}$ and $p_2=e^{\lambda/2}$.
On subsequent updates the prices keep interchanging,
so there is no convergence.
\end{example}

In addition, and unsurprisingly, as one approaches linear settings,
the step size employed by the tatonnement algorithm
needs to be increasingly small,
which leads to a slower rate of convergence,
and indicates a lack of robustness in the tatonnement procedure.

In this paper, we show that in suitable large Fisher markets,
this lack of robustness disappears, so long as approximate
rather than exact convergence suffices.
In addition, we obtain fast, i.e.\ linear, convergence
to an approximate equilibrium.
To see why approximate convergence is a reasonable and even
the right goal, consider dynamical settings;
in these settings, the equilibrium state can be expected to
change over time, and then the natural convergence
question becomes how closely can one track the
moving equilibrium? The answer is that it is a function of the
rate of change and the market parameters, as analyzed
by Cheung  et al.~\cite{cheung2019tracing}.
Clearly, in this type of setting, similar results will arise
with an approximate convergence result.

Our large market assumption requires that for goods with high elasticity,
price changes cause a relatively small change in spending.
In the case of buyers with linear utility functions,
where the elasticity parameters are unbounded,
this occurs if the buyers are heterogeneous, meaning the collection of their
utility functions is quite varied; also, we need each
individual buyer to constitute a small portion of the market
and for the number of buyers to be large compared to the number of goods.

To explain the intuition behind our results, we recall that Cheung, Cole and Devanur~\cite{cheung2019tatonnement}
showed that for many types of economies, including those we consider here,
a suitable tatonnement update is equivalent to a form of mirror descent
on a suitable convex function (actually, mirror ascent on a concave function).
To achieve convergence with mirror descent, one needs the function
$F$ to have a bounded Lipschitz parameter $L$, namely that
\begin{align*}
||\nabla F(\bbp) - \nabla F(\bbq)|| \le L||\bbp -\bbq||,
\end{align*}
for any two price vectors $\bbp$ and $\bbq$.
The rate of convergence will depend inversely on $L$.
Our large market assumption ensures this property so long as
$||\bbp-\bbq||$ is not too small.

In addition, the boundedness of the Lipschitz parameter holds
only if the prices are bounded away from 0.
Prior analyses implicitly bounded this parameter by showing the prices
are bounded away from zero, though this bound depended
on the initial prices and the particulars of the market.
In this paper, we assume there are minimum or reserve prices
which provides an alternate way to implicitly bound this parameter.

Furthermore, 
to obtain a linear rate of convergence one needs the function $F(p)$ to be 
strongly convex w.r.t.\ the equilibrium point $\bbp^*$, namely:
\[
F(\bbp) - f(\bbp^*) \ge \innerprod{\nabla F(\bbp)}{\bbp - \bbp^*}
+ \alpha ||\bbp- \bbp^*||^2.
\]
Again, our large market assumption ensures this property so long as
$||\bbp-\bbp^*||$ is not too small.

\paragraph{Related work} 
Two natural dynamics have been studied in the context of Fisher markets, tatonnement and proportional response.

The stability of the tatonnement process has been considered to be one of the most fundamental issues in general equilibrium theory.
Hahn \cite{Hahn1982} provides a thorough survey on the topic,
and the textbook of Mas-Colell, Whinston and Green~\cite{MWG95} contains a good summary of the classic results.

The longstanding interpretation of tatonnement is that it is a method used by an auctioneer for iteratively updating prices,
followed by trading at the equilibrium prices once they are reached.
If trading is allowed as the price updating occurs, this is called a \emph{non-tatonnement} process.
In recent years, discrete versions of the (non)-tatonnement process have received increased attention.
Codenotti et al.~\cite{CMV05} considered a tatonnement-like process that required some coordination among different goods and showed polynomial time convergence for a class Fisher markets with weak gross substitutes (WGS) utilities.
Cole and Fleischer~\cite{cole2008fast} were the first to establish fast convergence for a truly distributed, asynchronous and discrete version of tatonnement, 
once again for a class of WGS Fisher markets.
The continued interest in the plausibility of tatonnement is also reflected
in some experiments by Hirota~\cite{hirota05},
which showed the predictive accuracy of tatonnement in a non-equilibrium trade setting.

Proportional Response, in contrast, is a buyer-oriented update, 
originally analyzed in an effort
to explain the behavior of peer-to-peer networks~\cite{wu2007proportional,zhang2011proportional}. 
Here, buyers update their spending in
proportion to the contribution each good makes to its current utility.
An $O(1/T)$
rate of convergence was shown in~\cite{birnbaum2011distributed} for Fisher markets with buyers having linear utilities, and
for the substitutes domain excluding linear utilities, a faster linear rate of
convergence was shown in~\cite{zhang2011proportional}.
Recently, these results were generalized in~\cite{cheung2018dynamics} to the complete domain of CES utilities.

We note that the fastest rate of convergence for linear utilities in prior work
was $O(1/T)$. In contrast, the present work provides a linear rate of convergence
to an approximate equilibrium for linear utilities (and indeed for the full set of
utilities being considered, namely all CES utilities excluding Leontief utilities).

\paragraph{Roadmap}
In Section~\ref{sec:prelim}
we review standard definitions and notation, and follow this with the
statement of our result.
In Section~\ref{sec::progress} we provide a high level outline of our analysis,
which is expanded on in the following two sections.
Finally, in Section~\ref{sec::discussion} we mention some related open problems.
Some proofs and subsidiary lemmas are deferred to the appendix.

\section{Preliminaries}
\label{sec:prelim}

We use bold symbols, e.g., $\bbp,\bbx,\bbe$, to denote vectors.

\paragraph{Fisher Market}
In a Fisher market, there are $n$ perfectly divisible goods and $m$ buyers.
Without loss of generality, the supply of each good $j$ is normalized to be one unit.
Each buyer $i$ has a utility function $u_i:\rrplusn \ra \rr$, and a budget of size $e_i$.
At any given price vector $\bbp\in \rrplusn$, each buyer purchases a maximum utility
affordable collection of goods.
More precisely, $\bbx_i\in \rrplusn$ is said to be a demand of buyer $i$ if $~\bbx_i~\in~\argmax_{\bbx':~\bbx'\cdot \bbp \le e_i}~u_i(\bbx')$.

A price vector $\bps \in \rrplusn$ is called a \emph{market equilibrium} if at $\bps$, for every buyer $i$ there is a demand $\bbx_i$ such that
$$
\ps_j > 0~~\Rightarrow~~\sum_{i=1}^m x_{ij} ~=~ 1
\hsp\hsp\text{and}\hsp\hsp
\ps_j = 0~~\Rightarrow~~\sum_{i=1}^m x_{ij} ~\le~ 1.
$$
The collection of $\bbx_i$ is said to be an \emph{equilibrium allocation} to the buyers.

If there are reserve prices $\bbr$, then the prices are restricted
to the domain $\bbp \ge \bbr$, and in the second equilibrium condition,
$p_j^*=0$ is replaced by $p_j^* = r_j$.

\paragraph{CES utilities}
In this paper, each buyer $i$'s utility function is of the form
$$
u_i(\bbx_i) ~=~ \left(\sum_{j=1}^n a_{ij} \cdot (x_{ij})^\rhoi\right)^{1/\rhoi},
$$
for some $-\infty \le \rhoi \le 1$.
$u_i(\bbx_i)$ is called a Constant Elasticity of Substitution (CES) utility function.
They are a class of utility functions often used in economic analysis.
The limit as $\rhoi \ra -\infty$ is called a Leontief utility, usually written as
$u_i(\bbx_i) = \min_j \frac{\xij}{c_{ij}}$ \footnote{Here, the utility function $u_i(\bbx) = \min_j \frac{\xij}{c_{ij}}$ can be seen as the limit of $u_i(\bbx) = \Big( \sum_j \left(\frac{\xij}{c_{ij}}\right)^{\rho_i}\Big)^{\frac{1}{\rho_i}}$ as $\rho_i$ tends to $-\infty$.}; and the limit as $\rhoi \ra 0$ is called a Cobb-Douglas utility, 
usually written as $\prod_j {\xij}^{a_{ij}}$, with $\sum_j \aij = 1$. 
The utilities with $\rhoi\ge 0$ capture goods that are substitutes, 
and those with $\rhoi \le 0$ goods that are
complements. 
It is sometimes convenient to write $c_i  = \frac{\rho_i}{\rho_i - 1}$.

\paragraph{Notation} Buyer $i$'s spending on good $j$, denoted by $\bij$, is given by 
$\bij = \xij \cdot \pj$.
$E = \sum_i e_i$ denotes the total spending available to all the buyers.
$\zj = \sum_i \xij -1$ denotes the excess demand for good $j$.
$\kappa$ bounds the worst case ratio of the equilibrium and reserve prices:
$\kappa \ge \max_j p_j^*/r_j$.
We sometimes index prices, spending, and demands by $t$ to indicate the relevant value at time $t$.
Finally, we use a superscript of $^*$ to indicate an equilibrium value.

Our large market assumption states that for the buyers with linear
or close to linear utilities, the spending on a single good does not vary
too much as prices change.
We define ``close to linear'' in terms of a bound $\sigma>0$ on the
$\rho_i$ parameters.

\begin{assume}\label{ass::eps::market}[Large Market Assumption]
There is a (small) constant $\epsilon > 0$ such that
for those buyers with parameter $\rho_i \ge \sigma$,
\begin{align*}
\sum_{i: \rho_i \geq \sigma} | b_{ij}^t - b_{ij}^{t+1} |  \leq \epsilon \sum_i b_{ij}^t + \epsilon r_j.
\end{align*}
In addition, the total available money $E \ge \max_j r_j$.
\end{assume}
\paragraph{Remark}We validate our assumption in the following two settings. In the first setting the market has only a few buyers with $\rho_i$ bigger than $\sigma$. In this case, it's easy to see that the assumption holds if we set $\epsilon = \max \left\{ \frac{\sum_{i: \rho_i \geq \sigma} e_i}{r_j}\right\}$.

Our second setting is a large linear market. 
The property we want is that for each good $j$, when there are price
changes by factors of at most $e^{\pm\lambda}$, a relatively
small weight of buyers will be added to and removed from those 
currently purchasing good $j$ (where the weight is measured in terms
of the buyers' budgets.)

More specifically, $b_{ij}^t$ differs from $b_{ij}^{t+1}$ only if one of the following occur:
\begin{itemize}
\item there exists a $k$ such that $\frac{a_{ij}}{p_j^t} \leq \frac{a_{ik}}{p_k^t}$ and $\frac{a_{ij}}{p_j^{t+1}} \geq \frac{a_{ik}}{p_k^{t+1}}$
\item there exists a $k$ such that $\frac{a_{ij}}{p_j^t} \geq \frac{a_{ik}}{p_k^t}$ and $\frac{a_{ij}}{p_j^{t+1}} \leq \frac{a_{ik}}{p_k^{t+1}}$.
\end{itemize}
Note that  our price update rule ensures that $\frac{p_j^{t+1}}{p_k^{t+1}} \in [e^{-2\lambda}, e^{2\lambda}] \frac{p_j^{t}}{p_k^{t}}$. Therefore,
$b_{ij}^t$ differs from $b_{ij}^{t+1}$ only if there exists a $k$ such that $\frac{a_{ij}}{a_{ik}} \in [e^{-2\lambda}, e^{2\lambda}]\frac{p_j^{t}}{p_k^{t}}$. Also, since one of $b_{ij}^t$ and $b_{ij}^{t+1}$ is non-zero,  
for all $s$, $\frac{ a_{ij}}{a_{is}} \geq \frac{p_j^{t}}{p_s^{t}} e^{-2\lambda}$. 
Let $q_s \triangleq \frac{p_j^{t}}{p_s^{t}}$.
We conclude that 
\begin{align}
\label{eqn::changing-buyers}
\sum_i |b_{ij}^t - b_{ij}^{t+1}| \leq \sum_{i:\left\{ \substack{\exists k: \frac{ a_{ij}}{a_{ik}} \in \left[e^{-\lambda} q_k, e^{\lambda} q_k \right]\\ \text{~and~} \forall s~\left(\frac{ a_{ij}}{a_{is}} \geq q_s e^{-\lambda}\right)}\right.} e_i,
\end{align}
and
\begin{align}
\label{eqn::buyers-of-j}
\sum_i b_{ij}^t \geq \sum_{i: 
\forall s\ne j~\left(\frac{ a_{ij}}{a_{is}} > q_s\right)} e_i.
\end{align}

If the buyers are diverse, meaning that for any pair of goods, $j$ and $k$ say,
the ratios $\frac{a_{ij}}{a_{ik}}$ vary substantially across the buyers,
then so long as there are many buyers satisfying the condition 
$\forall s\ne j~\left(\frac{ a_{ij}}{a_{is}} > q_s\right)$ 
in \eqref{eqn::buyers-of-j},
it seems reasonable that their purchasing power be much larger than
that of the buyers satisfying the condition 
$\exists k: \frac{ a_{ij}}{a_{ik}} \in \left[e^{-\lambda} q_k, e^{\lambda} q_k \right] \text{~and~} \forall s~\left(\frac{ a_{ij}}{a_{is}} \geq q_s e^{-\lambda}\right)$ in \eqref{eqn::changing-buyers}.
While if there are few buyers satisfying the first condition, then
it is reasonable to assume that only a small number of buyers will
switch their purchase to or from good $j$, and that this changed
spending will be much smaller than $r_j$.

This motivates setting $\epsilon$ to be greater than or equal to
\begin{align*}
\max_{j, \bbq: q_k \in [\frac{r_k}{E}, \frac{E}{r_k}]}
\Bigg\{ 
\sum_{i: \left\{ \substack{ \exists k:~\frac{ a_{ij}}{a_{ik}} 
 \in \left[e^{-\lambda} q_k, e^{\lambda} q_k \right] \\
 \text{~and~} \forall s~\left(\frac{ a_{ij}}{a_{is}} \geq q_s e^{-\lambda}\right) } \right. 
}e_i \left/
\sum_{i:~\forall s\left(\frac{ a_{ij}}{a_{is}} > q_s\right)} e_i + r_j \right.
\Bigg\},
\end{align*}
\hide{
\begin{align*}
\max_{j, \bbq: q_k \in [\frac{r_k}{E}, \frac{E}{r_k}]}
\left\{ 
\frac
{\sum_{i:~\left\{ \substack{ \exists k: \frac{ a_{ij}}{a_{ik}} 
 \in \left[e^{-\lambda} q_k, e^{\lambda} q_k \right] \\
 \text{~and~} \forall s~\left(\frac{ a_{ij}}{a_{is}} \geq q_s e^{-\lambda}\right) } \right. 
}e_i}
{\sum_{i:~\forall s~\left(\frac{ a_{ij}}{a_{is}} > q_s\right)} e_i + r_j}
\right\},
\end{align*}
}
causing our assumption to hold.

\hide{Note that $$\sum_{i,k: \frac{ a_{ij}}{a_{ik}} \in [e^{-\lambda} q_k, e^{\lambda} q_k] \text{~and~} \forall s~\left(\frac{ a_{ij}}{a_{is}} \geq q_s e^{-\lambda}\right)} e_i  \leq  \sum_{i: e^{-\lambda} \max_s \{q_s a_{is}\} \leq a_{ij} \leq e^{\lambda} \max_s \{ q_s a_{is} \}} e_i$$ and $\sum_{i: 
\forall s~\left(\frac{ a_{ij}}{a_{is}} > q_s\right)} e_i = \sum_{i: a_{ij} \geq  \max_s \{ q_s a_{is}\}} e_i$. Therefore, so long as the total endowments of buyers with $e^{-\lambda} \max_s \{q_s a_{is}\} \leq a_{ij} \leq e^{\lambda} \max_s \{ q_s a_{is} \}$ is much smaller than the total endowments of buyers with $a_{ij} > \max_s \{ q_s a_{is}\}$ or $r_j$, we can pick a small $\epsilon$ for which our assumption holds. }

\vspace*{0.3in}

Our analysis is carried out with respect to the following potential function, 
which is the dual of the Eisenberg-Gale convex program:
\begin{align*}
\bbF(\bbp) = \sum_j p_j + \sum_i e_i \log \max_{\bbx_i \cdot \bbp = e_i} u_i(x_i)
\end{align*}
using the tatonnement update rule:
\begin{align*}
p_j^{t+1} = p_j^t \cdot e^{\Delta_j^t},
\end{align*}
where $\Delta_j^t = \lambda\min\{z_j^t,1\}$ and $\lambda \leq 1$, unless this update
would reduce $p_j^{t+1}$ below the reserve price, in
which case ${\Delta_j^t}$ is chosen so that $p_j^{t+1} = r_j$.

Our main result shows an initial linear rate of convergence toward
the equilibrium, and also shows that so long as the current prices
are not too close to the equilibrium then there is good progress toward
the equilibrium. The latter statement can also be viewed as a result
regarding the tracking of a slowly moving equilibrium.

Before we state the main result, we define a parameter $C(\kappa)$ introduced
in~\cite{cheung2013tatonnement}. Here $\kappa$ is an upper bound on the
ratio $\max_j \frac{p_j^*}{r_j}$.
$C(\kappa) = \min\left\{\frac{h_c(\kappa)}{c}, \frac{\kappa - 1 - \log \kappa}{(\kappa - 1)^2}\right\}$, where $h_c(\kappa) = \frac{1 - \kappa^c + c(\kappa-1)}{(\kappa - 1)^2}$ for any $\kappa \geq 0$ except $\kappa = 1$, and $h_c(1) = \frac{c(1 - c)}{2}$ and $c = \max_i c_i$. Note that $c<1$.

\begin{theorem}
\label{thm::main-result}
For any $0 < \theta < 1$, if $\frac {\lambda \sigma} {1 - \sigma} \le 1$
and $\kappa \ge \max_j \frac{p_j^*}{r_j}$, then
\begin{align*}
\bbF(\bbp^{t}) - \bbF(\bbp^*) \leq (1 - \alpha)^t \left( \bbF(\bbp^{0}) - \bbF(\bbp^*)\right) + 2 \frac{\lambda \epsilon^2 \mathcal{M}}{\alpha \theta},
\end{align*}
where $\alpha =  \frac{\left(1 - \lambda - 2 \lambda \cdot \max\left\{ \frac{\sigma}{1 - \sigma}, 1 \right\} - 2 \epsilon - 2 \theta\right)}{\max_j \left\{\max \left\{2, \frac{1}{2 C(\kappa)} \right\} \frac{E}{\lambda r_j}\right\}}$ and $\mathcal{M} =  \max \left\{ \sum_j p_j^0, \left(\left(e^\lambda - 2 \lambda \right) \frac{1 + 2 \lambda - e^\lambda}{\lambda} + \lambda\right) \left(E + \sum_j r_j\right) \right\}$.
Furthermore,
if $\bbF(\bbp^{t}) - \bbF(\bbp^*) \ge \frac{4 \lambda \epsilon^2 \mathcal{M}}{\alpha\theta}$ then
\begin{align*}
\bbF(\bbp^{t+1}) - \bbF(\bbp^*) \leq \left(1 -\frac  {\alpha}{2}\right) \left( \bbF(\bbp^{t}) - \bbF(\bbp^*)\right).
\end{align*}
\end{theorem}

In Section~\ref{sec::dyn-markets}, we show a variant of this theorem which demonstrates that in  dynamical settings, i.e.\ setting in which the equilibrium point changes over time, but not too quickly, this changing equilibrium can be tracked using tatonnement updates.
\section{A High Level Overview of the Analysis}
\label{sec::progress}

The analysis is largely based on deriving two bounds.
The first is a progress lemma,
a lower bound on the reduction in value of $F(p)$ due to the
time $t$ update, stated in Lemma~\ref{lem::prog} below.
The second is an upper bound on the distance to the equilibrium, stated
in Lemma~\ref{lem::dis::opt} below.
We also need to relate the sum of the prices at time $t+1$ to the
corresponding sum for time $0$, as stated in Lemma~\ref{lem::sum::price}.
Our main result then follows fairly readily.

With a slight abuse of notation, we let $u_i(\bbb_i, \bbp)$ denote buyer $i$'s
utility when spending $\bbb_i$ at prices $\bbp$. 

We analyze the change to the potential function due to the time $t$ updates.
Note that since buyers best respond, $\max_{x_i \cdot p = e_i} u_i(x_i) = u_i(\bbb_i^{t}, \bbp^{t})$. So,
\begin{align*}
\bbF(\bbp^{t+1}) - \bbF(\bbp^t) = \sum_j (p_j^{t+1} - p_j^t) + \sum_i e_i \log \frac{u_i(\bbb_i^{t+1}, \bbp^{t+1})}{u_i(\bbb_i^{t}, \bbp^{t})}. \numberthis \label{eqn::1}
\end{align*}

\begin{lemma} \label{lem::prog}
For any $0 < \sigma < 1$ such that $\frac{\lambda \sigma}{1 - \sigma} \leq 1$, if $| \Delta_j^t | \leq \lambda |\min\{z_j^t, 1\}|$
and $\sign(\Delta_j^t) = \sign(\min\{z_j^t, 1\})$, then
\begin{align*}
\bbF(\bbp^t) - \bbF(\bbp^{t+1})   \geq&  \left(1 - \lambda - 2 \lambda \cdot \max\left\{ \frac{\sigma}{1 - \sigma}, 1 \right\}\right) \sum_j p_j^t z_j^t \Delta_j^t 
- \sum_{i: \rho_i \geq \sigma} \rho_i \sum_j \left(b_{ij}^{t} - b_{ij}^{t+1} \right) \Delta_j^t.
\end{align*}
\end{lemma}

To obtain an upper bound on $\bbF(\bbp) - \bbF(\bbp^*)$ 
we follow the approach taken in~\cite{cheung2013tatonnement}.
They pointed out that if $\frac{p_j^*}{p_j} \leq \kappa$ for all $j$, then
\begin{align}
\label{eqn::strg-conv-bdd}
\bbF(\bbp^*) - \bbF(\bbp) - \innerprod{\nabla \bbF(\bbp)} {\bbp^* - \bbp}
\geq \sum_j C(\kappa) x_j \frac{(p_j^* - p_j)^2}{p_j},
\end{align}
where $C(\kappa)$ is a suitable parameter we specify later.
We note that this is a strong convexity bound of the type we need
for a linear convergence rate.

We will show:
\begin{lemma}\label{lem::dis::opt}
If $\kappa \ge \max_j \frac {p_j^*}{r_j}$, then
\begin{align*}
\bbF(\bbp^t) - \bbF(\bbp^*) \leq \sum_j \max \left\{2, \frac{1}{2 C(\kappa)} \right\} \frac{E}{\lambda r_j} p_j^t z_j^t \Delta_j^t.
\end{align*}
\end{lemma}

Finally, the bound on the sum of the prices is stated in the next lemma.
\begin{lemma}\label{lem::sum::price}
Using the definition of $\mathcal{M}$ from Theorem~\ref{thm::main-result} gives
\begin{align*}
\sum_j p_j^{t+1} & \le \max \left\{ \sum_j p_j^0, \left(\left(e^\lambda - 2 \lambda \right) \frac{1 + 2 \lambda - e^\lambda}{\lambda} + \lambda\right) \left(E + \sum_j r_j\right) \right\} = \mathcal{M}.
\end{align*}
\end{lemma}

We are now ready to prove our main result.
\begin{proof} [Proof of Theorem~\ref{thm::main-result}]
 We will be applying Lemma~\ref{lem::prog}, and we begin by bounding the second term on the RHS of the expression there.
\begin{align*}
\sum_{i : \rho_i \geq \sigma} \rho_i \sum_j \left(b_{ij}^t - b_{ij}^{t+1}\right) \Delta_j^t &\leq \sum_{i : \rho_i \geq \sigma} \rho_i \sum_j \left| b_{ij}^t - b_{ij}^{t+1} \right| \left|\Delta_j^t \right| \\
&\leq \sum_{i : \rho_i \geq \sigma} \sum_j \left| b_{ij}^t - b_{ij}^{t+1} \right| \left|\Delta_j^t \right|.
\end{align*}
By Assumption~\ref{ass::eps::market} for the first inequality, and because $p_j^t \geq r_j$ for the second inequality,

\begin{align*}
\sum_{i : \rho_i \geq \sigma} \rho_i \sum_j \left(b_{ij}^t - b_{ij}^{t+1}\right) \Delta_j^t &\leq \sum_j (\epsilon \sum_i b_{ij}^t + \epsilon r_j) \left| \Delta_j^t \right| = \sum_j (\epsilon p_j^t (1 + z_j^t) +\epsilon r_j) \left| \Delta_j^t \right| \\
&\leq \sum_j 2 \epsilon p_j^t \left| \Delta_j^t \right| + \epsilon p_j^t z_j^t \left| \Delta_j^t \right| \leq\sum_j 2 \epsilon p_j^t \left| \Delta_j^t \right| + \epsilon p_j^t z_j^t \Delta_j^t. \numberthis \label{ineq::error::ad2}
\end{align*}

We use the following result: for any $\theta > 0$,
\begin{align*}
\sum_j 2 \epsilon p_j^t \left| \Delta_j^t \right| \leq \sum_j 2 \theta p_j^t z_j^t \Delta_j^t + 2 \frac{\lambda \epsilon^2}{\theta} \sum_j p_j^t. \numberthis \label{ineq::error::ad}
\end{align*}

This holds because if $\epsilon p_j^t \left| \Delta_j^t \right| \geq \theta p_j^t z_j^t \Delta_j^t$, then
$\epsilon \geq \theta |z_j^t|$.
Therefore
\begin{align*}
2 \epsilon p_j^t \left| \Delta_j^t \right| \leq 2 \frac{\lambda\epsilon^2}{\theta} p_j^t.
\end{align*}

Substituting \eqref{ineq::error::ad2} and \eqref{ineq::error::ad} in Lemma~\ref{lem::prog} yields
\begin{align*}
\bbF(\bbp^t) - \bbF(\bbp^{t+1}) \geq &\left(1 - \lambda - 2 \lambda \cdot \max\left\{ \frac{\sigma}{1 - \sigma}, 1 \right\} - 2 \epsilon - 2\theta\right) \sum_j p_j^t z_j^t \Delta_j^t 
-  2 \frac{\lambda \epsilon^2}{\theta} \sum_j p_j^t
\end{align*}

Applying Lemma~\ref{lem::sum::price} gives
\begin{align*}
\bbF(\bbp^t) - \bbF(\bbp^{t+1}) \geq &\left(1 - \lambda - 2 \lambda \cdot \max\left\{ \frac{\sigma}{1 - \sigma}, 1 \right\} - 2 \epsilon - 2\theta\right) \sum_j p_j^t z_j^t \Delta_j^t 
-  2 \frac{\lambda \epsilon^2}{\theta} \mathcal{M}.
\numberthis  \label{ineq::used::later}\end{align*}
Applying Lemma~\ref{lem::dis::opt} and recalling that
$\alpha =  \frac{\left(1 - \lambda - 2 \lambda \cdot \max\left\{ \frac{\sigma}{1 - \sigma}, 1 \right\} - 2 \epsilon - 2\theta\right)}{\max_j \left\{\max \left\{2, \frac{1}{2 C(\kappa_j)} \right\} \frac{E}{\lambda r_j}\right\}}$ yields:
\begin{align*}
\bbF(\bbp^t) - \bbF(\bbp^{t+1}) \geq \alpha[F(\bbp^t) - F(\bbp^*)] 
-  2 \frac{\lambda \epsilon^2}{\theta} \mathcal{M}.
\end{align*}
Our first claim follows readily:
\begin{align}
\label{eqn::one-step-progress}
\bbF(\bbp^{t+1}) - \bbF(\bbp^*) &\leq (1 - \alpha) \left(\bbF(\bbp^t) - \bbF(\bbp^*) \right) + 2 \frac{\lambda \epsilon^2}{\theta} \mathcal{M} \\
\nonumber
&\leq (1 - \alpha)^{\rjc{t}} \left( \bbF(\bbp^{0}) - \bbF(\bbp^*)\right) + 2 \frac{\lambda \epsilon^2 \mathcal{M}}{\theta}\left(1 + (1 - \alpha) + (1 - \alpha)^2 + \cdots\right) \\
\nonumber
&\leq (1 - \alpha)^{\rjc{t}} \left( \bbF(\bbp^{0}) - \bbF(\bbp^*)\right) + 2 \frac{\lambda \epsilon^2 \mathcal{M}}{\alpha \theta}.
\end{align}

To prove the second claim, recall that we are assuming
$\bbF(\bbp^{t}) - \bbF(\bbp^*)
\leq 4 \frac{\lambda \epsilon^2}{\alpha\theta} \mathcal{M}$.
Then, by \eqref{eqn::one-step-progress},
\begin{align*}
\bbF(\bbp^{t+1}) - \bbF(\bbp^*) &\leq (1 - \alpha) \left(\bbF(\bbp^t) - \bbF(\bbp^*) \right) +\frac {\alpha}{2} \bbF(\bbp^{t}) - \bbF(\bbp^*)\\
&\le (1 - \frac{\alpha}{2}) \left(\bbF(\bbp^t) - \bbF(\bbp^*) \right).
\end{align*}
\end{proof}

\section{The Proof of Lemma~\ref{lem::prog}, the Progress Lemma}
\label{sec::proof-progress-lemma}

The starting point for our analysis is \eqref{eqn::1}.
The first step is to bound 
$\log \frac{u_i(\bbb^{t+1}_i, \bbp^{t+1}_i)}{u_i(\bbb^t_i, \bbp^t_i)}$.
The next four lemmas provide a variety of bounds
depending on the value of $\rho_i$ and other parameters.

\begin{lemma} \label{lem::linear}
If buyer $i$ has a linear utility function, then
\begin{align*}
e_i \log \frac{u_i(\bbb^{t+1}_i, \bbp^{t+1}_i)}{u_i(\bbb^t_i, \bbp^t_i)} \leq - \sum_j b_{ij}^t \Delta_j^t + \sum_j (b_{ij}^t - b_{ij}^{t+1})\Delta_{j}^t.
\end{align*}
\end{lemma}

\begin{lemma}For any $0 < \rho_i < 1$,  if $| \Delta_j^t | \leq 1$ for all $j$ and $t$, then
\begin{align*}
e_i \log \frac{u_i(\bbb_i^{t+1}, \bbp^{t+1})}{u_i(\bbb_i^t, \bbp^t)} \leq  - \sum_j b_{ij}^t \Delta_j^t + \sum_j b_{ij}^t \rho_i \left(\Delta_j^t\right)^2 - \rho_i  \sum_j b_{ij}^{t+1} \Delta_j^t + \rho_i \sum_j b_{ij}^t \Delta_j^t.
\end{align*} \label{lem::2}
\end{lemma}

\begin{lemma} \label{lem::1}
For any $\rho_i > 0$, if $|\lambda c_i | \leq 1$ and $| \Delta_j^t | \leq 1$ for all $j$ and $t$, then 
\begin{align*}
 e_i \log \frac{u_i(\bbb_i^{t+1}, \bbp^{t+1})}{u_i(\bbb_i^{t}, \bbp^{t})}  \leq - \sum_j b_{ij}^t \Delta_j^t - \sum_j b_{ij}^t c_i \left(\Delta_j^t\right)^2.
\end{align*}
\end{lemma}

\begin{lemma} \label{lem::comp}
If buyer $i$ has a complementary utility function, then
\begin{align*}
e_i \log \frac{u_i(\bbb^{t+1}_i, \bbp^{t+1}_i)}{u_i(\bbb^t_i, \bbp^t_i)} \leq - \sum_j b_{ij}^t \Delta_j^t.
\end{align*}
\end{lemma}

We are now ready to prove Lemma~\ref{lem::prog}.

\pfof{Lemma~\ref{lem::prog}}
Recall that $c_i = \frac{\rho_i}{\rho_i - 1}$, and $\sigma$ is a threshold designating the buyers to which Assumption~\ref{ass::eps::market}
applies, namely those with $\rho_i \ge \sigma$. 
We apply Lemma~\eqref{lem::1} to the buyers with $0 < \rho_i \leq \sigma$. 
In order to apply Lemma~\eqref{lem::1}, it suffices to have
\begin{align*}
\lambda \cdot \frac{\sigma}{ 1 - \sigma }\leq 1.
\end{align*}
Therefore, by Lemmas~\ref{lem::linear}--\ref{lem::comp}
and equation~\eqref{eqn::1}, for any $0 < \sigma < 1$ such that $\lambda \frac{\sigma}{1 - \sigma} \leq 1$,
\begin{align*}
\bbF(\bbp^{t+1}) - \bbF(\bbp^t) &= \sum_j p_j (e^{\Delta_j^t}  - 1) - \sum_{ij} b_{ij}^t \Delta_j^t \\
&\hsp\hsp\hsp\hsp - \sum_{ij: 0 < \rho_i < \sigma} b_{ij}^t c_i (\Delta_j^t)^2 + \sum_{ij: \sigma \leq \rho_i < 1} b_{ij}^t \rho_i (\Delta_j^t)^2  \\
&\hsp\hsp\hsp\hsp + \sum_{i: \rho_i \geq \sigma} \rho_i \sum_j \left(b_{ij}^{t} - b_{ij}^{t+1} \right) \Delta_j^t \\
&= \sum_j p_j(e^{\Delta_j^t} - \Delta_j^t - 1) - \sum_j (\sum_i b_{ij}^t - p_j^t) \Delta_j^t \\
&\hsp\hsp\hsp\hsp - \sum_{ij: 0 < \rho_i < \sigma} b_{ij}^t c_i (\Delta_j^t)^2 + \sum_{ij: \sigma \leq \rho_i < 1} b_{ij}^t \rho_i (\Delta_j^t)^2  \\
&\hsp\hsp\hsp\hsp + \sum_{i: \rho_i \geq \sigma} \rho_i \sum_j \left(b_{ij}^{t} - b_{ij}^{t+1} \right) \Delta_j^t.
\end{align*}

Note that $e^{\Delta_j^t} - \Delta_j^t - 1 \leq \left(\Delta_j^t\right)^2$ as $| \Delta_j^t | \leq 1$, and $\sum_i b_{ij}^t - p_j^t = p_j^t z_j^t$. Therefore,
\begin{align*}
\bbF(\bbp^{t+1}) - \bbF(\bbp^t) &\leq \underbrace{\sum_j p_j \left( \Delta_j^t \right)^2}_{A} - \underbrace{\sum_j p_j^t z_j^t \Delta_j^t}_{B} \\
&\hsp\hsp\hsp\hsp \underbrace{- \sum_{ij: \rho_i < \sigma} b_{ij}^t c_i (\Delta_j^t)^2}_{C} + \underbrace{\sum_{ij: \sigma \leq \rho_i < 1} b_{ij}^t \rho_i (\Delta_j^t)^2}_{D}  \\
&\hsp\hsp\hsp\hsp + \sum_{i: \rho_i \geq \sigma} \rho_i \sum_j \left(b_{ij}^{t} - b_{ij}^{t+1} \right) \Delta_j^t.
\end{align*}

It is easy to see that $A \leq \lambda B$ if $| \Delta_j^t | \leq \lambda |\min\{z_j^t, 1\}|$ and $\sign(\Delta_j^t) = \sign(\min\{z_j^t, 1\})$. 
Next, we will bound $C$ and $D$ in terms of $B$. 
To this end, we note that we can omit the portion of $C$ with $\rho_i \le 0$
as for these terms $c_i \ge 0$ and consequently removing them
only increases the RHS expression. We then note that for $\rho_i>0$,
\[
-c_i =- \frac {\rho_i}{\rho_i -1} \le -\frac{\sigma}{\sigma-1} = \frac{\sigma}{\sigma-1}.
\]
For term $D$ we use the simple bound $\rho_i \le 1$.
Thus terms $C$ and $D$ are bounded by
\[
\sum_{i,j: \rho_i >0} \max\left\{\frac{\sigma}{\sigma-1},1\right\} b_{ij}^t 
\left(\Delta_j^t\right)^2.
\]

We now give a bound on this expression in terms of $B$.

\begin{claim}\label{lem::lower::progress} If $| \Delta_j^t | \leq \lambda |\min\{z_j^t, 1\}|$ and $\sign(\Delta_j^t) = \sign(\min\{z_j^t, 1\})$, then
\begin{align*}
\sum_j p_j^t z_j^t \Delta_j^t \geq \frac{1}{2 \lambda} \left(\sum_i b_{ij}^t\right) (\Delta_j^t)^2. 
\end{align*}
\end{claim}

Thus
\begin{align*}
\bbF(\bbp^{t}) - \bbF(\bbp^{t+1}) &\geq 
\sum_j p_j^t z_j^t \Delta_j^t 
\left(1 - \lambda - 2\lambda \max\left\{\frac{\sigma}{\sigma-1},1\right\}\right).
\end{align*}

\end{proof}

\section{Bounding the Distance to the Optimum}
\label{sec:dist-to-opt}

In this section, we provide an upper bound on $\bbF(\bbp) - \bbF(\bbp^*)$.  

\hide{
We follow the approach taken in~\cite{cheung2013tatonnement}.
They pointed out that if $\frac{p_j^*}{p_j} \leq \kappa_j$ for all $j$, then
\begin{align*}
\bbF(\bbp^*) - l_{\bbF}(\bbp^*; \bbp) \geq \sum_j C(\kappa) x_j \frac{(p_j^* - p_j)^2}{p_j},
\end{align*}
where $C(\kappa) = \max\left\{\frac{h_c(\kappa)}{c}, \frac{\kappa - 1 - \log \kappa}{(\kappa - 1)^2}\right\}$, $h_c(\kappa) = \frac{1 - \kappa^c + c(\kappa-1)}{(\kappa - 1)^2}$ for any $\kappa \geq 0$ except $\kappa = 1$, and $h_c(1) = \frac{c(1 - c)}{2}$ and $c = \max_i c_i$. 

Therefore, 
}

Equation \eqref{eqn::strg-conv-bdd} yields
\begin{align*}
\bbF(\bbp^t) - \bbF(\bbp^*) &\leq \sum_j z_j^t (p_j^* - p_j^t) - \sum_j C(\kappa_j) x^t_j \frac{(p_j^* - p_j^t)^2}{p_j^t} \\
&\leq \max_{\bbp' \geq \bbr} \sum_j \left(z_j^t (p'_j - p_j^t) -  C(\kappa_j) x^t_j \frac{(p'_j - p_j^t)^2}{p_j^t} \right). \numberthis \label{ineq::dist::opt::1}
\end{align*}

In \cite{cheung2013tatonnement}, they prove that 
\begin{align*}
\max_{p'} \left(z_j^t (p'_j - p_j^t) -  C(\kappa_j) x^t_j \frac{(p'_j - p_j^t)^2}{p_j^t} \right) \leq \max \left\{ 2, \frac{1}{2 C(\kappa_j)} \right\} \left(z_j^t \right)^2 p_j^t. \numberthis \label{ineq::dist::opt::2}
\end{align*}

Here, we want to prove one more upper bound.
\begin{lemma} \label{lem::linear::dis}
If $p_j^{t+1} = r_j$ and $z_j^t \leq 0$, then 
\begin{align*}
\max_{p'_j \geq r_j}  \left(z_j^t (p'_j - p_j^t) -  C(\kappa_j) x^t_j \frac{(p'_j - p_j^t)^2}{p_j^t} \right) \leq z_j^t p_j^t \log \frac{p_j^{t+1}}{p_j^t}.
\end{align*}
\end{lemma}
It's easy to see that our update rule has $\Delta_j^t = \log \frac{p_j^{t+1}}{p_j^t}$.
Now, combining \eqref{ineq::dist::opt::1}, \eqref{ineq::dist::opt::2}, and Lemma~\ref{lem::linear::dis}, yields Lemma~\ref{lem::dis::opt}.
\hide{
\begin{lemma}\label{lem::dis::opt}
\begin{align*}
\bbF(\bbp^t) - \bbF(\bbp^*) \leq \max \left\{2, \frac{1}{2 C(\kappa)} \right\} \frac{E}{\lambda r_j} p_j^t z_j^t \Delta_j^t.
\end{align*}
\end{lemma}
}

\pfof{Lemma~\ref{lem::linear::dis}}
Let 
\begin{align*}
\mathcal{I}(p'_j) = z_j^t (p'_j - p_j^t) -  C(\kappa_j) x^t_j \frac{(p'_j - p_j^t)^2}{p_j^t}
\end{align*}
and in the setting without reserve prices, let
\begin{align*}
p_j^{opt} = \max_{p'_j} \left\{\mathcal{I}(p'_j) \right\}.
\end{align*}

In this case, the maximum value of $\mathcal{I}(p')$ is 
\begin{align*}
\frac{\left(z_j^t\right)^2 p_j^t}{4 C(\kappa_j) x_j},
\end{align*}
and the optimum value is
\begin{align*}
p_j^{opt} = p_j^t + \frac{z_j^t p_j^t}{2 C(\kappa_j) x_j^t}.
\end{align*}

\noindent
{\bf Case 1}: $p_j^{opt} \geq r_j$.
\\
This implies 
\begin{align*}
p_j^t + \frac{z_j^t p_j^t}{2 C(\kappa_j) x_j^t} \geq r_j.
\end{align*}
Therefore, since $z_j^t \leq 0$,
\begin{align*}
\frac{\left(z_j^t\right)^2 p_j^t}{4 C(\kappa_j) x_j} \leq \frac{1}{2} z_j^t p_j^t \frac{r_j - p_j^t}{p_j^t} \leq \frac{1}{2} z_j^t p_j^t \log \frac{r_j}{p_j^t}.
\end{align*}

\noindent
{\bf Case 2}: $p_j^{opt} < r_j$. \\
Note that $\mathcal{I}(p'_j)$ is a quadratic function. 
On the domain $p'_j \ge r_j$, it achieves its maximum value when $p'_j = r_j$. 
Therefore, the maximum value is
\begin{align*}
\frac{\left(z_j^t\right)^2 p_j^t}{4 C(\kappa_j) x_j} - \frac{C(\kappa_j) x_j^t}{p_j^t} (r_j - p_j^{opt})^2 &= \frac{\left(z_j^t\right)^2 p_j^t}{4 C(\kappa_j) x_j} - \frac{C(\kappa_j) x_j^t}{p_j^t} \left(r_j - p_j^t - \frac{z_j^t p_j^t}{2 C(\kappa_j) x_j^t}\right)^2 \\
& = - C(\kappa_j) x_j^t \frac{(r_j - p_j^t)^2}{p_j^t} + z_j^t(r_j - p_j^t)
\end{align*}
Since $C(\kappa_j) x_j^t \frac{(r_j - p_j^t)^2}{p_j^t} \geq 0$, this is less than
\begin{align*}
z_j^t(r_j - p_j^t) \leq z_j^t p_j^t \frac{r_j - p_j^t}{p_j^t} \leq z_j^t p_j^t \log \frac{r_j}{p_j^t}.
\end{align*}
\end{proof}

\pfof{Lemma~\ref{lem::dis::opt}}
{\bf Case 1}: $p_j^{t+1} > r_j^t$.\\
Then $\Delta_j^t = \lambda \min\{z_j^t, 1\}$. Note that $z_j^t  = \frac{\sum_i b_{ij}^t - p_j^t}{p_j^t} \leq \frac{E}{r_j}$. Therefore, $| \Delta_j^t | \geq \frac{\lambda r_j}{E} |z_j^t|$. Since $\Delta_j^t$ and $z_j^t$ are both positive or both negative, combining with \eqref{ineq::dist::opt::1} and \eqref{ineq::dist::opt::2} yields
\begin{align*}
\bbF(\bbp^t) - \bbF(\bbp^*) \leq \max \left\{2, \frac{1}{2C(\kappa_j)}\right\} (z_j^t)^2 p_j^t \leq \max\left\{2, \frac{1}{2C(\kappa_j)}\right\} \frac{E}{\lambda r_j} p_j^t z_j^t \Delta_j^t.
\end{align*}

\noindent
{\bf Case 2}: $p_j^{t+1} = r_j^t$.\\
 By \eqref{ineq::dist::opt::1} and using Lemma~\ref{lem::linear::dis}, the result follows as $\Delta_j^t = \log \frac{p_j^{t+1}}{p_j^t}$, 
and $E \ge r_j$ by assumption.
\end{proof}
\section{Dynamical Markets}
\label{sec::dyn-markets}

In this section, we will consider dynamical market. For each round, there can be a small change to the supplies, budgets and buyers' preferences.
We will seek to show that tatonnement can cause the prices to
pursue the market equilibrium. 
Note that, in general, we need to modify the potential function to account for the possibly changing supplies $w_j^t$ for item $j$ at time $t$; the new potential function is
\begin{align*}
\sum_j w_j^t p^t_j + \sum_i e_i \log \max_{\bbx_i \cdot \bbp^t = e_i^t} u_i(x_i),
\end{align*}
and our update rule will be 
\begin{align*}
p_j^{t+1} = p_j^t e^{\Delta_j^t}
\end{align*}
where $\Delta_j^t = \lambda \max \left\{\frac{z_j^t}{w_j} , 1\right\}$.

Cheung, Hoefer, and Nakhe \cite{cheung2019tracing} analyze the following settings:
\begin{itemize}
\item \emph{Supply Change} If at time $t$, the supplies change by at most 
$\epsilon$, then by at most $(P + E) \epsilon$, where $P$ is the maximum price at time $t+1$;
\item \emph{Budget Change} If at time $t$, the sum of the absolute values of the changes to the buyers' budgets is at most $\epsilon$, then the potential function changes by at most $C \epsilon$, where $C$ is the maximum possible ratio between a buyer's utility at time $t+1$ and her utility at the market equilibrium at time $t+1$;
\item \emph{Utility Change} If at time $t$, given any prices, the  ratio of the utility difference when best responding is bounded by $\chi$, then the potential function changes by at most $2E \chi$.
\end{itemize}
In order to analyze the effect of these changes over time, in this paper, we let $D$ denote the maximum change to the potential function at each round. 
We let $\bbp^{t,*}$ denote the equilibrium prices at time $t$.
We have the following theorem.
\begin{theorem}
\label{thm::second-result}
For any $0 < \theta < 1$, if $\frac {\lambda \sigma} {1 - \sigma} \le 1$
and $\kappa \ge \max_j \frac{p_j^*}{r_j}$, then
\begin{align*}
\bbF^{t+1}(\bbp^{t+1}) - \bbF^{t+1}(\bbp^{t+1,*}) \leq (1 - \alpha)^t \left( \bbF(\bbp^{0}) - \bbF(\bbp^{0,*})\right) + \frac{1}{\alpha} \left(2 \frac{\lambda \epsilon^2 \mathcal{M}}{ \theta} + D\right),
\end{align*}
where $\alpha =  \frac{\left(1 - \lambda - 2 \lambda \cdot \max\left\{ \frac{\sigma}{1 - \sigma}, 1 \right\} - 2 \epsilon - 2 \theta\right)}{\max_j \left\{\max \left\{2, \frac{1}{2 C(\kappa)} \right\} \frac{E}{\lambda r_j}\right\}}$ and $\mathcal{M} =  \max \left\{ \sum_j \hat{w}_j p_j^0, \left(\left(e^\lambda - 2 \lambda \right) \frac{1 + 2 \lambda - e^\lambda}{\lambda} + \lambda\right) \left(E + \sum_j r_j\right) \right\}$. Also, $E$ here will be the maximum possible total money over time and $\hat{w}_j$ will be the maximum supply of item $j$ over time.
Furthermore,
if $\bbF^t(\bbp^{t}) - \bbF^t(\bbp^{t,*}) \ge \frac{2}{\alpha} \left(\frac{2 \lambda \epsilon^2 \mathcal{M}}{\theta} + D \right)$ then
\begin{align*}
\bbF^{t+1}(\bbp^{t+1}) - \bbF^{t+1}(\bbp^{t+1,*}) \leq \left(1 -\frac  {\alpha}{2}\right) \left( \bbF^t(\bbp^{t}) - \bbF^t(\bbp^*)\right).
\end{align*}
\end{theorem}
\begin{proof}
This theorem follows directly if we replace $2 \frac{\lambda \epsilon^2 \mathcal{M}}{ \theta}$ by $2 \frac{\lambda \epsilon^2 \mathcal{M}}{ \theta} + D$ in \eqref{ineq::used::later}.
\end{proof}

\section{Discussion and Open Questions}
\label{sec::discussion}

The strong convexity of the function $F(p)$
reduces as complementarity increases and disappears when Leontief
utility functions are allowed.
Is there a suitable large market assumption that will create strong
convexity in the large when these utility functions are present?
Also, do the current results extend to asynchronous updating
as used in the Ongoing Market model of Cole and Fleischer~\cite{cole2008fast}
and elsewhere~\cite{cheung2012tatonnement,cheung2018dynamics}?

Do analogous results hold for Proportional Response?
Note that for CES utility functions, there is a convex function
on which mirror descent corresponds to the Proportional Response
update~\cite{cheung2018dynamics}, 
and this function is strongly convex away from the
extremes of linear and Leontief utility functions (actually, the
situation is more complicated; the function is a mix of concave and convex).

\newpage

\newpage

\bibliographystyle{plain}
\bibliography{sample}

\begin{thebibliography}{10}

\bibitem{birnbaum2011distributed}
Benjamin Birnbaum, Nikhil~R. Devanur, and Lin Xiao.
\newblock Distributed algorithms via gradient descent for {F}isher markets.
\newblock In {\em Proceedings of the 12th ACM conference on Electronic
  commerce}, pages 127--136. ACM, 2011.

\bibitem{chen2009settling}
Xi~Chen, Decheng Dai, Ye~Du, and Shang-Hua Teng.
\newblock Settling the complexity of {A}rrow-{D}ebreu equilibria in markets
  with additively separable utilities.
\newblock In {\em 2009 50th Annual IEEE Symposium on Foundations of Computer
  Science}, pages 273--282. IEEE, 2009.

\bibitem{chen2013complexity}
Xi~Chen, Dimitris Paparas, and Mihalis Yannakakis.
\newblock The complexity of non-monotone markets.
\newblock In {\em Proceedings of the forty-fifth annual ACM symposium on Theory
  of computing}, pages 181--190. ACM, 2013.

\bibitem{cheung2018amortized}
Yun~Kuen Cheung and Richard Cole.
\newblock Amortized analysis of asynchronous price dynamics.
\newblock In {\em 26th European Symposium on Algorithms, ESA 2018}. Schloss
  Dagstuhl-Leibniz-Zentrum fur Informatik GmbH, Dagstuhl Publishing, 2018.

\bibitem{cheung2013tatonnement}
Yun~Kuen Cheung, Richard Cole, and Nikhil Devanur.
\newblock Tatonnement beyond gross substitutes? {G}radient descent to the
  rescue.
\newblock In {\em Proceedings of the forty-fifth annual ACM symposium on Theory
  of computing}, pages 191--200. ACM, 2013.

\bibitem{cheung2019tatonnement}
Yun~Kuen Cheung, Richard Cole, and Nikhil~R. Devanur.
\newblock Tatonnement beyond gross substitutes? {G}radient descent to the
  rescue.
\newblock {\em Games and Economic Behavior}, 2019.

\bibitem{cheung2012tatonnement}
Yun~Kuen Cheung, Richard Cole, and Ashish Rastogi.
\newblock Tatonnement in ongoing markets of complementary goods.
\newblock In {\em Proceedings of the 13th ACM Conference on Electronic
  Commerce}, pages 337--354. ACM, 2012.

\bibitem{cheung2018dynamics}
Yun~Kuen Cheung, Richard Cole, and Yixin Tao.
\newblock Dynamics of distributed updating in {F}isher markets.
\newblock In {\em Proceedings of the 2018 ACM Conference on Economics and
  Computation}, pages 351--368. ACM, 2018.

\bibitem{cheung2019tracing}
Yun~Kuen Cheung, Martin Hoefer, and Paresh Nakhe.
\newblock Tracing equilibrium in dynamic markets via distributed adaptation.
\newblock In {\em Proceedings of the 18th International Conference on
  Autonomous Agents and MultiAgent Systems}, pages 1225--1233. International
  Foundation for Autonomous Agents and Multiagent Systems, 2019.

\bibitem{CMV05}
Bruno Codenotti, Benton McCune, and Kasturi~R. Varadarajan.
\newblock Market equilibrium via the excess demand function.
\newblock In {\em STOC}, pages 74--83, 2005.

\bibitem{cole2008fast}
Richard Cole and Lisa Fleischer.
\newblock Fast-converging tatonnement algorithms for one-time and ongoing
  market problems.
\newblock In {\em Proceedings of the fortieth annual ACM symposium on Theory of
  computing}, pages 315--324. ACM, 2008.

\bibitem{devanur2004spending}
Nikhil~R. Devanur.
\newblock The spending constraint model for market equilibrium: algorithmic,
  existence and uniqueness results.
\newblock In {\em Proceedings of the thirty-sixth annual ACM symposium on
  Theory of computing}, pages 519--528. ACM, 2004.

\bibitem{duan2015combinatorial}
Ran Duan and Kurt Mehlhorn.
\newblock A combinatorial polynomial algorithm for the linear {A}rrow--{D}ebreu
  market.
\newblock {\em Information and Computation}, 243:112--132, 2015.

\bibitem{garg2017settling}
Jugal Garg, Ruta Mehta, Vijay~V. Vazirani, and Sadra Yazdanbod.
\newblock Settling the complexity of {L}eontief and {PLC} exchange markets
  under exact and approximate equilibria.
\newblock In {\em Proceedings of the 49th Annual ACM SIGACT Symposium on Theory
  of Computing}, pages 890--901. ACM, 2017.

\bibitem{garg2018strongly}
Jugal Garg and L{\'a}szl{\'o}~A. V{\'e}gh.
\newblock A strongly polynomial algorithm for linear exchange markets.
\newblock {\em arXiv preprint arXiv:1809.06266}, 2018.

\bibitem{ghiyasvand2012simple}
Mehdi Ghiyasvand and James~B. Orlin.
\newblock A simple approximation algorithm for computing {A}rrow-{D}ebreu
  prices.
\newblock {\em Operations research}, 60(5):1245--1248, 2012.

\bibitem{Hahn1982}
Frank Hahn.
\newblock Stability.
\newblock volume~2 of {\em Handbook of Mathematical Economics}, pages 745 --
  793. Elsevier, 1982.

\bibitem{hirota05}
Masayoshi Hirota, Ming Hsu, Charles~R. Plott, and Brian~W. Rogers.
\newblock Divergence, closed cycles and convergence in {S}carf environments:
  Experiments in the dynamics of general equilibrium systems.
\newblock Working Papers 1239, California Institute of Technology, Division of
  the Humanities and Social Sciences, October 2005.

\bibitem{jain2007polynomial}
Kamal Jain.
\newblock A polynomial time algorithm for computing an {A}rrow--{D}ebreu market
  equilibrium for linear utilities.
\newblock {\em SIAM Journal on Computing}, 37(1):303--318, 2007.

\bibitem{MWG95}
Andreu Mas-Collel, Michael~D. Whinston, and Jerry~R. Green.
\newblock {\em Microeconomic Theory}.
\newblock Oxford University Press, 1995.

\bibitem{orlin2010improved}
James~B. Orlin.
\newblock Improved algorithms for computing {F}isher's market clearing prices.
\newblock In {\em Proceedings of the forty-second ACM symposium on Theory of
  computing}, pages 291--300. ACM, 2010.

\bibitem{shmyrev2009algorithm}
Vadim~I. Shmyrev.
\newblock An algorithm for finding equilibrium in the linear exchange model
  with fixed budgets.
\newblock {\em Journal of Applied and Industrial Mathematics}, 3(4):505, 2009.

\bibitem{varian74}
Hal~R Varian.
\newblock Equity, envy, and efficiency.
\newblock {\em Journal of Economic Theory}, 9(1):63 -- 91, 1974.

\bibitem{wu2007proportional}
Fang Wu and Li~Zhang.
\newblock Proportional response dynamics leads to market equilibrium.
\newblock In {\em Proceedings of the thirty-ninth annual ACM symposium on
  Theory of computing}, pages 354--363. ACM, 2007.

\bibitem{ye2008path}
Yinyu Ye.
\newblock A path to the {A}rrow--{D}ebreu competitive market equilibrium.
\newblock {\em Mathematical Programming}, 111(1-2):315--348, 2008.

\bibitem{zhang2011proportional}
Li~Zhang.
\newblock Proportional response dynamics in the {F}isher market.
\newblock {\em Theoretical Computer Science}, 412(24):2691--2698, 2011.

\end{thebibliography}

\newpage

\appendix
\section{Missing Proofs}
\label{sec::missing-proofs}

\pfof{Lemma~\ref{lem::sum::price}}
\begin{align*}
\sum_j p_j^{t+1} &= \sum_j p_j^t e^{\Delta_j^t} 
= \sum_j p_j^t \left(e^{\Delta_j^t} - 1 - \Delta_j^t\right) + \sum_j p_j^t \left( 1 + \Delta_j^t \right).
\end{align*}
If $p_j^{t+1} = r_j$, then $p_j^t \left( 1 + \Delta_j^t \right) \leq p_j^t e^{\Delta_j^t} = r_j \leq (1 - \lambda) p_j^t + \lambda r_j$. Otherwise, $p_j^t \left( 1 + \Delta_j^t \right) \leq p_j^t \left(1 + \lambda z_j^t\right) = (1 - \lambda) p_j^t + \lambda \sum_i b_{ij}^t$. This implies
\begin{align*}
\sum_j p_j^{t+1} \leq \sum_j p_j^t \left(e^{\Delta_j^t} - 1 - \Delta_j^t\right) + \lambda \sum_j r_j + (1 - \lambda) \sum_j p_j^t + \lambda \sum_{ij} b_{ij}^t.
\end{align*}
Since $|\Delta_j^t| \leq \lambda$, $e^{\Delta_j^t} - 1 - \Delta_j^t \leq \max \left\{ e^\lambda - 1  - \lambda, e^{-\lambda} - 1 + \lambda \right\} \leq e^\lambda - 1  - \lambda$.
\begin{align*}
\sum_j p_j^{t+1} \leq \left( e^\lambda - 1  - \lambda\right) \sum_j p_j^t + \lambda\sum_j r_j + (1 - \lambda) \sum_j p_j^t + \lambda \sum_{ij} b_{ij}^t. \numberthis \label{ineq::p::e::1}
\end{align*}
If $\sum_j p_j^t \geq \left(\frac{1 + 2 \lambda - e^\lambda}{\lambda} \right)\left(\sum_{ij} b_{ij}^t + \sum_j r_j \right)$, then rearranging \eqref{ineq::p::e::1} gives $\sum_j p_j^{t+1} \leq \sum_j p_j^t$. Otherwise, replacing $\sum_j p_j^t$ by $\left(\frac{1 + 2 \lambda - e^\lambda}{\lambda} \right)\left(\sum_{ij} b_{ij}^t + \sum_j r_j \right)$ gives\\
$\sum_j p_j^{t+1} \leq \left(\left(e^\lambda - 2 \lambda \right) \frac{1 + 2 \lambda - e^\lambda}{\lambda} + \lambda\right) \left(E + \sum_j r_j\right)$.
Thus \\ 
$\sum_j p_j^{t+1}  \le \max \left\{ \sum_j p_j^t, \left(\left(e^\lambda - 2 \lambda \right) \frac{1 + 2 \lambda - e^\lambda}{\lambda} + \lambda\right) \left(E + \sum_j r_j\right) \right\}$
and the result follows by induction on $t$.
\end{proof}

\pfof{Lemma~\ref{lem::linear}}
For simplicity, we can assume that at any given time each buyer will spend all her money on just one item. To handle the general case, we partition each  buyer into several buyers, each of whom buys one good. Then the same result follows.

So, here we use $j(i, \bbp)$ to denote the item with max utility-per-dollar for buyer $i$ at price $\bbp$ and buyer $i$ spends the whole budget on this item. Note that 
\begin{align*} 
e_i \log \frac{u_i(\bbb^{t+1}_i, \bbp^{t+1}_i)}{u_i(\bbb^t_i, \bbp^t_i)} &=  e_i \log \frac{a_{ij(i, \bbp^{t+1})}}{p_{j(i, \bbp^{t+1})}^{t+1}} -  e_i \log \frac{a_{ij(i, \bbp^{t})}}{p_{j(i, \bbp^{t})}^{t}} \\
&=   e_i \log \frac{a_{ij(i,\bbp^{t+1})} p^{t+1}_{j(i, \bbp^t)}}{a_{ij(i, \bbp^t)} p^{t+1}_{j(i, \bbp^{t+1})}}  -  b_{ij(i, \bbp^t)}^t \log \frac{p^{t+1}_{j(i, \bbp^t)}}{p^{t}_{j(i, \bbp^t)}} \\
&= e_i \log \frac{a_{ij(i,\bbp^{t+1})} p^{t+1}_{j(i, \bbp^t)}}{a_{ij(i, \bbp^t)} p^{t+1}_{j(i, \bbp^{t+1})}} -  b_{ij(i, \bbp^t)}^t \Delta_{j(i, \bbp^t)}^t. \numberthis \label{eqn::linear::1}
\end{align*}
We also know that 
\begin{align*}
\frac{a_{ij(i, \bbp^t)}}{p^{t+1}_{j(i, \bbp^t)}}  = \frac{a_{ij(i, \bbp^t)}}{p^{t}_{j(i, \bbp^t)} e^{\Delta_{j(i, \bbp^t)}^t}}  \geq  \frac{a_{ij(i, \bbp^{t+1})}}{p^{t}_{j(i, \bbp^{t+1})} e^{\Delta_{j(i, \bbp^t)}^t}} = \frac{a_{ij(i, \bbp^{t+1})} e^{\Delta_{j(i, \bbp^{t+1})}^t}}{p^{t+1}_{j(i, \bbp^{t+1})} e^{\Delta_{j(i, \bbp^t)}^t}}.
\end{align*}
Therefore, 
\begin{align*}
 e_i \log \frac{a_{ij(i, \bbp^{t+1})} p^{t+1}_{j(i, \bbp^t)}}{a_{ij(i, \bbp^t)} p^{t+1}_{j(i, \bbp^{t+1})}} \leq  e_i (\Delta_{j(i, \bbp^t)}^t-  \Delta_{j(i, \bbp^{t+1})}^t) = \sum_j (b_{ij}^t - b_{ij}^{t+1})\Delta_{j}^t. \numberthis \label{eqn::linear::2}
\end{align*}
To see the final equality, note that $b_{ij(i,\bbp^t)}= e_i$ and $b_{ij}=0$ for 
all other $j$;
so $\sum_j b_{ij}^t \Delta_j^t = e_i \Delta^t_{j(i,\bbp^t)}$;
likewise, $\sum b_{ij}^{t+1}\Delta_j^t = e_i \Delta^t_{j(i,\bbp^{t+1})}$.

Combining \eqref{eqn::linear::1} and \eqref{eqn::linear::2} yields the result.
\end{proof}

The remaining lemmas use the following observations from  \cite{cheung2013tatonnement}.
\begin{align*}
\max_{\bbx_i \cdot \bbp = e_i} u_i(x_i) = \left\{ \begin{array}{lr} e_i \left(\sum_j a_{ij}^{1 - c_i} p_j^{c_i} \right)^{- \frac{1}{c_i}} ~~~~~&\mbox{$\rho_i < 1$} \\ e_i \frac{a_{ij}}{p_j} ~~~~~&\mbox{$\rho_i  = 1$}\end{array} \right. , \numberthis \label{eqn::max::utility}
\end{align*}
And the best response to price $\bbp$ is
\begin{align*}
b_{ij} = e_i \frac{a_{ij}^{1 - c_i} p_j^{c_i}}{\sum_{j'} a_{ij'}^{1 - c_i} p_{j'}^{c_i}}. \numberthis \label{eq::best::resp}
\end{align*}

\pfof{Lemma~\ref{lem::2}}
First, we decompose the LHS into two parts:
\begin{align*}
e_i \log \frac{u_i(\bbb_i^{t+1}, \bbp^{t+1})}{u_i(\bbb_i^t, \bbp^t)} = e_i \log \frac{u_i(\bbb_i^{t+1}, \bbp^{t+1})}{u_i(\bbb_i^t, \bbp^{t+1})}  + e_i \log \frac{u_i(\bbb_i^{t}, \bbp^{t+1})}{u_i(\bbb_i^t, \bbp^t)}. 
\end{align*}
We start by bounding the first term. Note that 
\begin{align*}
u_i(\bbb_i, \bbp) = \left(\sum_j a_{ij} \left(\frac{b_{ij}}{p_j}\right)^{\rho_i} \right)^{\frac{1}{\rho_i}}.
\end{align*}

Using \eqref{eq::best::resp} yields:
\begin{align*}
e_i \log \frac{u_i(\bbb_{i}^{t+1}, \bbp^{t+1})}{u_i(\bbb_{i}^{t}, \bbp^{t + 1})}   &= \frac{e_i}{\rho_i} \log \dfrac{ \sum_j a_{ij} \left(\frac{e_i \frac{a_{ij}^{1 - c_i} (p_j^{t+1})^{c_i}}{\sum_{j'} a_{ij'}^{1 - c_i} (p_{j'}^{t+1})^{c_i}}}{p_j^{t+1}}\right)^{\rho_i}}{\sum_j a_{ij} \left(\frac{e_i \frac{a_{ij}^{1 - c_i} (p_j^{t})^{c_i}}{\sum_{j'} a_{ij'}^{1 - c_i} (p_{j'}^{t})^{c_i}}}{p_j^{t + 1}}\right)^{\rho_i}}\\
& = \frac{e_i}{\rho_i} \log  \frac{\sum_j a_{ij} \left(\frac{a_{ij}^{1 - c_i} (p_j^{t + 1})^{c_i}}{p_j^{t+1}}\right)^{\rho_i}}{\sum_j a_{ij} \left(\frac{a_{ij}^{1 - c_i} (p_j^{t })^{c_i}}{p_j^{t+1}}\right)^{\rho_i}} + e_i \log \frac{\sum_{j'} a_{ij'}^{1 - c_i} (p_{j'}^{t})^{c_i}}{\sum_{j'} a_{ij'}^{1 - c_i} (p_{j'}^{t+1})^{c_i}}.
\end{align*}

By calculation, $1 + (1 - c_i) \rho_i = 1 - c_i$ and $(c_i - 1) \rho_i = c_i$. So,
\begin{align*}
\sum_j a_{ij} \left(\frac{a_{ij}^{1 - c_i} (p_j^{t + 1})^{c_i}}{p_j^{t+1}}\right)^{\rho_i} = \sum_j a_{ij}^{1 - c_i} \left(p_j^{t+1}\right)^{c_i} ~~~~~~\mbox{and} \\
\sum_j a_{ij} \left(\frac{a_{ij}^{1 - c_i} (p_j^{t })^{c_i}}{p_j^{t+1}}\right)^{\rho_i} = \sum_j a_{ij}^{1 - c_i} \left(p_j^{t}\right)^{c_i} \left( \frac{p_j^t}{p_j^{t+1}}\right)^{\rho_i}
\end{align*}

Therefore,
\begin{align*}
e_i \log \frac{u_i(\bbb_{i}^{t+1}, \bbp^{t+1})}{u_i(\bbb_{i}^{t}, \bbp^{t + 1})}  & = e_i \frac{1 - \rho_i}{\rho_i} \log \frac{ \sum_j a_{ij}^{1 - c_i} \left( p_j^{t+1} \right)^{c_i} }{ \sum_j a_{ij}^{1 - c_i} \left( p_j^t \right)^{c_i} \left( \frac{p_j^t}{p_j^{t+1}} \right)^{\rho_i} } + e_i \log \frac{\sum_j a_{ij}^{1 - c_i} \left( p_j^t \right)^{c_i}}{\sum_j a_{ij}^{1 - c_i} \left( p_j^t \right)^{c_i} \left(\frac{p_j^t}{p_j^{t+1}} \right)^{\rho_i}}.
\end{align*}
Note that, by \eqref{eq::best::resp}, $b_{ij}^{t+1} = e_i \frac{a_{ij}^{1 - c_i} \left(p^{t+1}_j\right)^{c_i}}{\sum_{j'} a_{ij'}^{1 - c_i} \left(p^{t+1}_{j'}\right)^{c_i}}$ and $b_{ij}^{t} = e_i \frac{a_{ij}^{1 - c_i} \left(p^{t}_j\right)^{c_i}}{\sum_{j'} a_{ij'}^{1 - c_i} \left(p^{t}_{j'}\right)^{c_i}}$.
Thus
\begin{align*}
&e_i \log \frac{u_i(\bbb_{i}^{t+1}, \bbp^{t+1})}{u_i(\bbb_{i}^{t}, \bbp^{t + 1})} \\
& \smallspace = - e_i \frac{1 - \rho_i}{\rho_i} \log \sum_j \frac{b_{ij}^{t+1}}{e_i} \left( \frac{p_j^t}{p_j^{t+1}} \right)^{\rho_i} \left( \frac{p_j^t}{p_j^{t+1}} \right)^{c_i} - e_i \log \sum_j \frac{b_{ij}^t}{e_i} \left(\frac{p_j^t}{p_j^{t+1}} \right)^{\rho_i}.
\end{align*}
As the $\log$ function is concave, $\log \sum_i a_i x_i \geq \sum_i a_i \log x_i$ when $\sum_i a_i = 1$; this yields:
\begin{align*}
e_i \frac{1 - \rho_i}{\rho_i} \log \sum_j \frac{b_{ij}^{t+1}}{e_i} \left( \frac{p_j^t}{p_j^{t+1}} \right)^{\rho_i} \left( \frac{p_j^t}{p_j^{t+1}} \right)^{c_i} &\geq \sum_j b_{ij}^{t+1} ( \rho_i + c_i) \frac{1 - \rho_i}{\rho_i} \log \frac{p_j^t}{p_j^{t+1}} \\
&= -\rho_i \sum_j b_{ij}^{t+1} \log \frac{p_j^t}{p_j^{t+1}}\\
\mbox{and}\largespace e_i \log \sum_j \frac{b_{ij}^t}{e_i} \left(\frac{p_j^t}{p_j^{t+1}} \right)^{\rho_i}   &\geq \rho_i \sum_j b_{ij}^t \log \frac{p_j^t}{p_j^{t+1}}.
\end{align*}
Therefore,
\begin{align*}\label{ineq::close::1}
e_i \log \frac{u_i(\bbb_{i}^{t+1}, \bbp^{t+1})}{u_i(\bbb_{i}^{t}, \bbp^{t + 1})} & \leq  \rho_i \sum_j b_{ij}^{t+1} \log \frac{p_j^t}{p_j^{t+1}} - \rho_i \sum_j b_{ij}^t \log \frac{p_j^t}{p_j^{t+1}} \\
& = - \rho_i  \sum_j b_{ij}^{t+1} \Delta_j^t + \rho_i \sum_j b_{ij}^t \Delta_j^t. \numberthis
\end{align*}

Now let's look at the second part, $e_i \log \frac{u_i(\bbb_i^{t}, \bbp^{t+1})}{u_i(\bbb_i^t, \bbp^t)}$.
\begin{align*}
e_i \log \frac{u_i(\bbb_{i}^{t}, \bbp^{t+1})}{u_i(\bbb_{i}^{t}, \bbp^{t})}   &= \frac{e_i}{\rho_i} \log \dfrac{ \sum_j a_{ij} \left(\frac{e_i \frac{a_{ij}^{1 - c_i} (p_j^{t})^{c_i}}{\sum_{j'} a_{ij'}^{1 - c_i} (p_{j'}^{t})^{c_i}}}{p_j^{t+1}}\right)^{\rho_i}}{\sum_j a_{ij} \left(\frac{e_i \frac{a_{ij}^{1 - c_i} (p_j^{t})^{c_i}}{\sum_{j'} a_{ij'}^{1 - c_i} (p_{j'}^{t})^{c_i}}}{p_j^{t}}\right)^{\rho_i}}\\
&= \frac{e_i}{\rho_i} \log \dfrac{ \sum_j a_{ij} \left(\frac{a_{ij}^{1 - c_i} (p_j^{t})^{c_i}}{p_j^{t+1}}\right)^{\rho_i}}{\sum_j a_{ij} \left(\frac{a_{ij}^{1 - c_i} (p_j^{t})^{c_i}}{p_j^{t}}\right)^{\rho_i}}.
\end{align*}

Recall that  $1 + (1 - c_i) \rho_i = 1 - c_i$ and $(c_i - 1) \rho_i = c_i$. So,
\begin{align*}
e_i \log \frac{u_i(\bbb_{i}^{t}, \bbp^{t+1})}{u_i(\bbb_{i}^{t}, \bbp^{t})} = \frac{e_i}{\rho_i} \log \frac{\sum_j a_{ij}^{1 - c_i} \left(p_j^t\right)^{c_i} \left( \frac{p_j^t}{p_j^{t+1}} \right)^{\rho_i}}{\sum_j a_{ij}^{1 - c_i} \left(p_j^t\right)^{c_i}}.
\end{align*}

Remember that $b_{ij}^t = e_i \frac{a_{ij}^{1 - c_i} \left( p_j^t \right)^{c_i}}{\sum_{j'} a_{ij'}^{1 - c_i} \left( p_{j'}^t \right)^{c_i}}$. Therefore,
\begin{align*}
e_i \log \frac{u_i(\bbb_{i}^{t}, \bbp^{t+1})}{u_i(\bbb_{i}^{t}, \bbp^{t})} = \frac{e_i}{\rho_i} \log \sum_j \frac{b_{ij}^t}{e_i}  \left( \frac{p_j^t}{p_j^{t+1}} \right)^{\rho_i} \leq \frac{e_i}{\rho_i} \sum_j \frac{b_{ij}^t}{e_i} \left( \left( \frac{p_j^t}{p_j^{t+1}} \right)^{\rho_i} - 1\right),
\end{align*}
using the fact that $\log x \leq x - 1$ for the last inequality, and noting that $\sum_j b_{ij}^t = e_i$.

As $\frac{p_j^{t+1}}{p_j^t} = e^{\Delta_j^t}$ and $|\Delta_j^t | \leq 1$, 
\begin{align*}
\frac{ \left( \frac{p_j^t}{p_j^{t+1}} \right)^{\rho_i} - 1}{\rho_i} \leq - \Delta_j^t + \rho_i \left(\Delta_j^t\right)^2.
\end{align*}

Therefore,
\begin{align*} \label{ineq::close::2}
e_i \log \frac{u_i(\bbb_{i}^{t}, \bbp^{t+1})}{u_i(\bbb_{i}^{t}, \bbp^{t})} \leq - \sum_j b_{ij}^t \Delta_j^t + \sum_j b_{ij}^t \rho_i \left(\Delta_j^t\right)^2. \numberthis
\end{align*}

Combining \eqref{ineq::close::1} and \eqref{ineq::close::2} gives the result.
\end{proof}

The following claim is used in the final two results.
\begin{claim}
\label{clm::util-ratio-bdd}
If $\rho <1$,
\begin{align*}
 e_i \log \frac{u_i(\bbb_i^{t+1}, \bbp^{t+1})}{u_i(\bbb_i^{t}, \bbp^{t})} = 
 -\frac{e_i}{c_i} \log \sum_j \frac{b_{ij}^t}{e_i} \left( \frac{p_j^{t+1}}{p_j^t}\right)^{c_i}.
\end{align*}
\end{claim}
\begin{proof}
As $\rho_i < 1$, by \eqref{eqn::max::utility},  $\max_{x_i \cdot p = e_i} u_i(x_i) = e_i \left(\sum_j a_{ij}^{1 - c_i} p_j^{c_i} \right)^{- \frac{1}{c_i}}$. Thus,
\begin{align*}
 e_i \log \frac{u_i(\bbb_i^{t+1}, \bbp^{t+1})}{u_i(\bbb_i^{t}, \bbp^{t})} = -\frac{e_i}{c_i} \log \frac{\sum_j a_{ij}^{1 - c_i} \left(p_j^{t+1}\right)^{c_i}}{\sum_j a_{ij}^{1 - c_i} \left(p_j^{t} \right)^{c_i}}.
\end{align*}

Substituting from \eqref{eq::best::resp} gives
\begin{align*}
e_i \log \frac{u_i(\bbb_i^{t+1}, \bbp^{t+1})}{u_i(\bbb_i^{t}, \bbp^{t})} = -\frac{e_i}{c_i} \log \frac{\sum_j a_{ij}^{1 - c_i} \left(p_j^{t+1}\right)^{c_i}}{\sum_j a_{ij}^{1 - c_i} \left(p_j^{t} \right)^{c_i}} = -\frac{e_i}{c_i} \log \sum_j \frac{b_{ij}^t}{e_i} \left( \frac{p_j^{t+1}}{p_j^t}\right)^{c_i}.
\end{align*}
\end{proof}

\pfof{Lemma~\ref{lem::1}}
By applying Claim~\ref{clm::util-ratio-bdd}, and noting
 that $p_j^{t+1} = p_j^t e^{\Delta_j^t}$ and $| \Delta_j^t | \leq \lambda$ and  $c_i < 0$, yields
\begin{align*}
e_i \log \frac{u_i(\bbb_i^{t+1}, \bbp^{t+1})}{u_i(\bbb_i^{t}, \bbp^{t})} &= -\frac{e_i}{c_i} \log \sum_j \frac{b_{ij}^t}{e_i} \left( \frac{p_j^{t+1}}{p_j^t}\right)^{c_i} \\
&\leq -\frac{e_i}{c_i} \sum_j \frac{b_{ij}^t}{e_i} \left[\left( \frac{p_j^{t+1}}{p_j^t}\right)^{c_i} - 1 \right]\\
&\largespace \mbox{(using $\log x \leq x - 1$)} \\
&= - \frac{e_i}{c_i} \sum_j \frac{b_{ij}^t}{e_i} \left(e^{c_i \Delta_j^t} - 1 \right) \\
&\leq - \sum_j b_{ij}^t \frac{c_i \Delta_j^t + \left( c_i \Delta_j^t\right)^2}{c_i} \\
&\largespace \mbox{(using $e^x \leq 1 + x + x^2$ if $-1 \leq x = c_i \Delta_j^t \leq 1$)} \\
&= - \sum_j b_{ij}^t \Delta_j^t - \sum_j b_{ij}^t c_i \left(\Delta_j^t\right)^2.
\end{align*}
\end{proof}

\pfof{Lemma~\ref{lem::comp}}
As $\rho \le 0$, by Claim~\ref{clm::util-ratio-bdd},
\begin{align*}
e_i \log \frac{u_i(\bbb^{t+1}_i, \bbp^{t+1}_i)}{u_i(\bbb^t_i, \bbp^t_i)} = -\frac{e_i}{c_i} \log \sum_j \frac{b_{ij}^t}{e_i} \left(\frac{p_j^{t+1}}{p_j^t} \right)^{c_i} 
\leq - \frac{e_i}{c_i} \sum_j \frac{b_{ij}^t}{e_i} \log \left(\frac{p_j^{t+1}}{p_j^t} \right)^{c_i} 
\leq - \sum_j b_{ij}^t \Delta_j^t.
\end{align*}
The first inequality holds as $\log$ is a concave function, $c_i > 0$ for complementary buyers, and $\sum_j \frac{b_{ij}^t}{e_i} =1$;
the final inequality uses $\Delta_j^t = \log \frac{p_j^{t+1}}{p_j^t}$.
\end{proof}

\pfof{Claim~\ref{lem::lower::progress}}
If $z_j = -1$ then $\sum_j b_{ij}^t = 0$ and the claim holds.
Otherwise, $z_j > -1$ and
\begin{align*}
p_j^t z_j^t \Delta_j^t &= \left(\sum_i b_{ij}^t - p_j \right) \Delta_j^t = \left(\sum_j b_{ij}^t\right)  \left( 1 - \frac{1}{1 + z_j^t}\right) \Delta_j^t \\
&\geq \left(\sum_i b_{ij}^t\right)  \left( 1 - \frac{1}{1 + \frac{\Delta_j^t}{\lambda}}\right) \Delta_j^t.
\end{align*}
If $z_j^t \geq 0$, then $0 \leq \frac{\Delta_j^t}{\lambda} \leq 1$. This implies
\begin{align*}
1 - \frac{1}{1 + \frac{\Delta_j^t}{\lambda}} \geq \frac{\Delta_j^t}{2 \lambda};
\end{align*}
and if $z_j^t < 0$, then $-\lambda < \Delta_j^t \leq 0$. This implies
\begin{align*}
1 - \frac{1}{1 + \frac{\Delta_j^t}{\lambda}} \leq \frac{\Delta_j^t}{\lambda} < \frac{\Delta_j^t}{2\lambda}.
\end{align*}
The result now follows.
\end{proof}
\end{document}